\documentclass[cmp,envcountsect,envcountsame,envcountreset]{svjour}
\usepackage{amsfonts,amssymb,graphicx}
\usepackage{epstopdf}
\usepackage{cite}
\usepackage{indentfirst}
\usepackage{newlfont}

\usepackage{latexsym}
\usepackage{amsmath, amssymb, mathrsfs, dsfont, mathtools, empheq}
\usepackage{graphicx}
\usepackage{comment}
\usepackage{color}

\numberwithin{equation}{section}

\renewcommand{\d}{\mathrm d}
\newcommand{\x}{\hat{x}}
\newcommand{\xn}{\x_n}
\newcommand{\eps}{\varepsilon}
\newcommand{\poo}{p^{(0)}}

\newcommand{\Zcl}{Z^{(0)}_{N}}
\newcommand{\pcl}{p^{(0)}_{N}}
\newcommand{\mucl}{\mu^{(0)}_{N}}
\newcommand{\DD}{\mathscr D}
\newcommand{\R}{\mathbb R}
\newcommand{\MM}{\mathscr M}
\newcommand{\Prob}{\mathbb P}
\newcommand{\N}{\mathbb N}
\newcommand{\1}{\mathds 1}
\newcommand{\Z}{\mathbb Z}
\newcommand{\C}{\mathbb C}
\def\be{\begin{equation}}
\def\ee{\end{equation}}
\newcommand{\eq}{\begin{equation}}
\newcommand{\eeq}{\end{equation}}
\newcommand{\E}{\mathbb E}
\DeclareMathOperator*{\argmax}{argmax}
\DeclareMathOperator{\dist}{dist}

\journalname{Arxiv preprint, June 12th 2015}

\begin{document}

\title{Limit theorems for monomer-dimer mean-field models with attractive potential}

\author{Diego Alberici\and Pierluigi Contucci\and Micaela Fedele \and Emanuele Mingione}

\institute{
			University of Bologna - Department of Mathematics \\
            piazza di Porta San Donato 5, Bologna (Italy) \\
            \email{diego.alberici2@unibo.it, pierluigi.contucci@unibo.it, micaela.fedele2@unibo.it, emanuele.mingione2@unibo.it} }

\date{}


\maketitle

\begin{abstract}
The number of monomers, in a monomer-dimer mean-field model with an attractive potential,
fluctuates according to the central limit theorem when the parameters are outside the critical curve.
At the critical point the model belongs to the same universality class of the mean-field ferromagnet. Along the critical
curve the monomer and dimer phases coexist.
\end{abstract}


\section*{Introduction}
Interacting particle systems described with statistical mechanics models are known
to have different fluctuation properties on their critical points.
In the mean-field ferromagnet, for instance, the sum of the spins centered around its mean and
normalised with the square root of the total volume, converges toward a normal random variable (central
limit theorem) away from the critical line.
At the critical point instead a non-normal behaviour emerges, i.e. the limiting probability distribution
for the sum of the spins centered and suitably normalised is not Gaussian \cite{simon, ellis1978statistics}.

In this paper we consider a mean-field system of interacting monomers and dimers where,
beyond the hard-core interaction among particles, an attractive interaction is added to favour
configurations where similar particles lie in neighbouring sites. The peculiar features of the presented model come from
the combined presence of the two interactions.
We show that the central limit theorem and the law of large numbers hold for the number of monomers at general values of the parameters.
At the critical point instead the central limit theorem breaks down and the number of monomers centered around its mean and normalised with the exponent $3/4$ of the total volume has a limiting density proportional to $\exp(-cx^4)$, i.e. the system exhibits the same critical behavior of the mean-field ferromagnet.
We also show that along the critical curve the law of large numbers breaks down, due to the coexistence of the {\it monomer} and the {\it dimer} phases. Unlike the Curie-Weiss model, the relative weights of these phases are non-constant and display two contributions that correspond to the two types of interaction.

We provide a rigorous proof of the mentioned results by first studying the properties
of the moment generating function for the model when the attractive interaction is zero.
Here the difficulty of the problem stems from
the fact that even in the absence of the attraction the system keeps its interacting nature and
the equilibrium measure does not factorise. To solve this problem we use a Gaussian representation
for the pure monomer-dimer model previously introduced in \cite{ACM2} which has the purpose of {\it decoupling}
the hard-core interaction. When instead we consider the attractive potential we follow the Gaussian
convolution method introduced in \cite{ellis1980limit}.

It would be interesting to further extend the results presented in this paper in the same spirit
of those obtained for the mean-field ferromagnet in \cite{ellis1978limit,ellis1980limit}
and also test for the same purpose other methods like those based on interchangeability \cite{chatt, loewe}
or those of Lee-Yang type  \cite{Lebowitz}.

The paper is organised as follows. Section 1 presents the definition of the model and the precise statements of the results.
In Section 2 we consider the pure hard-core model and prove the law of large numbers and central limit theorem by using the Gaussian representation and an extended Laplace method (reported in the Appendix).
In Section 3 we consider the  hard-core model with attraction and, using the method of the Gaussian convolution together with the formerly introduced Gaussian representation, we prove the law of large numbers, the central
limit theorem and their breakdown respectively along the critical curve and at the critical point.

\section{Definitions and Results}\label{defmain}

Let $G=(V,E)$ be a finite graph with vertex set $V$ and edge set $E\subseteq\{uv\equiv\{u,v\}\;|\;u,v\in V,\,u\neq v\}$.

\begin{definition}
A dimer configuration on the graph $G$ is a set $D$ of pairwise non-incident edges, called dimers:
\begin{equation}\label{dconstraint}
D\subseteq E \quad\text{and}\quad
\text{for each } v\in V \text{ there is at most one } u\in V \text{ such that } uv\in D \;.  
\end{equation}
The associated set of dimer-free vertices, called monomers, is denoted by
\begin{equation}
\MM_G(D):=\{v\in V\;|\;\forall u\in V\ uv\not\in D\; \}.
\end{equation}
We denote by $\DD_G$ the configuration space, i.e. the set of all possible dimer configurations on the graph $G\,$.
\end{definition}

We notice that by definition
\be\label{hardcore}
|\MM_G(D)|+2|D|=|V| \quad\forall\,D\in\DD_G \;.
\ee

In this paper we restrict our attention to the complete graph with vertex set $V=\{1, \ldots, N\}$ and edge set $E=\{uv\,|\, u,v\in V,\, u\neq v\}\,$. The corresponding configuration space will be denoted by $\DD_N$ and the set of monomers associated to the dimer configuration $D$ by $\MM_N(D)$.

A fundamental quantity is the \textit{number of monomers} for a given dimer configuration $D\in\DD_N\,$: $S_N(D):=|\MM_N(D)|$.
$S_N(D)$ can be seen as a sum of $N$ variables introducing, for a given $D\in\DD_N$ and for all $v\in V$, a monomer occupancy variable
\begin{equation} \label{eq: alpha}
\alpha_v(D) := \begin{cases}
1\,,& \text{if }v\in\MM_N(D) \\
0\,,& \text{otherwise}
\end{cases}.
\end{equation}
Thus, one can write
\begin{equation}\label{mag}
S_{N}(D)= \sum_{v=1}^{N}\alpha_{v}(D).
\end{equation}
We also define the \textit{empirical monomer density} as
\begin{equation}\label{mag}
m_{N}(D):=\frac{1}{N}S_N(D)
\end{equation}
which represents an analogous  of the empirical magnetization in magnetic models.

In  \cite{ACM} the authors consider a monomer-dimer model with imitative interaction on the complete graph, that we call \textit{Imitative Monomer-Dimer model} (IMD model), defined as follows.
For each integer $N$,  the Hamiltonian function is
\begin{equation}\label{hami}
- H_{N}(D) \,:=\, N\left(\left(h-J\right)m_{N}(D)+J\,\left(m_{N}(D)\right)^2 \right) \quad\forall\,D\in \DD_N
\end{equation}
where $h\in\R$ is the \textit{external field} and $J\geq0$ is the \textit{imitative coupling}.
The Hamiltonian \eqref{hami} induces a Gibbs probability measure on the configurations space $\DD_N\,$
\begin{equation} \label{eq: mu}
\mu_{N}(D) \,:=\, \frac{1}{Z_N}\;N^{-|D|}\exp(-H_{N}(D)) \quad\forall\,D\in\DD_N \;,
\end{equation}
where
\begin{equation}\label{partitionf}
Z_N \,:=\, \sum_{D\in\DD_N}N^{-|D|}e^{-H_N(D)}
\end{equation}
is the \textit{partition function}.
The factor $N^{-|D|}$ is the necessary normalisation working on the complete graph.
As usual, the quantity
\be\label{pressuredens}
p_N \,:=\, \frac{1}{N}\log Z_N
\ee
is called \textit{pressure density}.

\begin{remark}
Despite the Hamiltonian \eqref{hami} depends only on the numbers of monomers, it is possible to show \cite{ACM} that in our case, namely on the complete graph, any general Hamiltonian depending also on the number of dimers and the relative couplings is equivalent, up to a constant, to \eqref{hami}. Thus, the parameters $h,J$ take into account both monomer/dimer external fields and monomer-monomer/dimer-dimer/monomer-dimer couplings.
\end{remark}

Beside the formal analogy between \eqref{hami} and the Hamiltonian function of a Curie-Weiss model, we want to stress their main difference: in the former the configuration space $\DD_N$ is \textbf{not} a product space because of the hard-core constraint \eqref{dconstraint}.

Let us briefly recall the results obtained in \cite{ACM}. We refer to the original work for the details.

\begin{theorem}\label{mainpressure}
The thermodynamic limit of the pressure density of the IMD model is given by
\begin{equation}\label{termo_limit}
 \lim_{N\rightarrow\infty}p_{N} \,=\, \sup_m \widetilde{p} (m)
\end{equation}
\noindent where
\begin{equation}\label{tildep}
\widetilde p(m)\,:=\,-Jm^2+\poo((2m-1)J+h) \;,
\end{equation}
\begin{equation}\label{funzione.p}
\poo(h):=-\frac{1-g(h)}{2}-\frac{1}{2}\log(1-g(h))=-\frac{1-g(h)}{2}-\log g(h)+ h \;,
\end{equation}
\begin{equation}\label{funzione.g}
g(h):=\frac{e^{h}}{2}(\sqrt{e^{2h}+4}-e^{h}) \;.
\end{equation}
\end{theorem}

\begin{remark}\label{genMF}
We notice that, in analogy with magnetic models, one can define a general mean-field Hamiltonian as
\be\label{meanfhami}
-H^{\textup{\tiny{MF}}}_N(D)=N\,f\big(m_N(D)\big)
\ee
for any bounded continuous function $f$. As in the case of spin mean-field models, using standard Large Deviations techniques, one can prove that
\be\label{boh}
\lim_{N\to\infty}\frac{1}{N}\log Z_N^{\textup{\tiny{MF}}}=\sup_{m}\big( f(m)-I(m)\big)
\ee
where the rate function $I$ is given by
\be\label{rate}
I(z)=\begin{cases}
z\log z+\frac{1-z}{2}\log(1-z) + \frac{1-z}{2}-\poo(0)& \text{if } z\in [0,1]\\
\infty & \text{otherwise}.
\end{cases}
\ee
\end{remark}

The aim of the present work is to describe the limiting distribution of the random variable $S_N$ with respect to
the measure $\mu_N$, in a suitable scaling when $N\to\infty\,$.
From now on $\delta_x$ is the Dirac measure centered at $x$, $\mathcal{N}\left(m,\sigma^2\right)$ denotes the normal distribution with mean $m$ and variance $\sigma^2$ and $\overset{\mathcal{D}}{\rightarrow}$ denotes the convergence in distribution with respect to the Gibbs measure $\mu_N$, as $N\to\infty\,$.

Let start by considering the case  $J=0$. The Hamiltonian \eqref{hami} at $J=0$ is a special case of the original problem considered by Heilmann and Lieb \cite{HL}. We introduce the following notation,
\be\label{parclass}
\Zcl\equiv Z_N\Big|_{J=0},\qquad\pcl\equiv p_{N}\Big|_{J=0},\qquad\mucl\equiv\mu_N\Big|_{J=0} \;.
\ee
Setting $J=0$ in Theorem \ref{mainpressure} we get
\begin{equation} \label{eq: p classic limit}
\lim_{N\to\infty}\pcl \,=\, \poo(h) \quad\forall\,h\in\R \;.
\end{equation}

Thus, the pressure is analytic as a function of $h$  and the unique value of the limiting monomer density is  given by
\be\label{freedensity}
\lim_{N\to\infty}\E_{\mucl}(m_{N}) \,=\, \lim_{N\to\infty}\frac{\partial}{\partial h}\pcl \,=\,
\frac{\partial}{\partial h}\poo
\ee
and, using the properties of $g$ given in \cite{ACM}, we get
\be\label{gprop}
\frac{\partial}{\partial h}\poo \,=\, g(h).
\ee

\begin{theorem}\label{teo1}
For the IMD model at $J=0$ the followings hold:
\begin{equation}\label{LLN0}
m_{N} \,\overset{\mathcal{D}}{\rightarrow}\, \delta_{g(h)}
\end{equation}
and
\begin{equation}\label{CLT0}
\frac{S_{N}-N\,g(h)}{\sqrt{N}} \,\overset{\mathcal{D}}{\rightarrow}\, \mathcal{N}\left(0,\,\frac{\partial}{\partial h}g(h)\right) \;.
\end{equation}
\end{theorem}

We notice that, even if we are in the case $J=0$, \eqref{CLT0} is not a  consequence
of the classical central limit theorem, indeed $S_N$ is not a sum of i.i.d. random variables because of the presence of the hard-core interaction. The proof of the theorem is in the next section.

Let us consider now the case $J>0$. It is possible to show \cite{ACM} that the points where the function $\widetilde p$ defined in \eqref{tildep} reaches its maximum satisfy the following consistency equation:
\begin{equation}\label{consequation}
m=g((2m-1)J+h).
\end{equation}
The analysis of  \eqref{consequation} allows to identify the region where there exists a unique global maximum point $m^*(h,J)$ of $\widetilde p$. The resulting picture (see figure \ref{fig.3.3})  is the following: the function $m^*$ is single-valued and continuous on the plane $(h,J)$ with the exception of an open curve $\gamma$ defined by an implicit equation $h=\gamma(J)$. Moreover $m^*$ is smooth outside $\gamma$ union its endpoint $(h_c,J_c)$.
Instead on $\gamma$, there are two global maximum points $m_1(J)= m_1(\gamma(J),J)$ and $m_2(J)=m_2(\gamma(J),J)$. In particular, $m_1<m_2$ thus they represents respectively the \textit{dimer} and the \textit{monomer phase}. The curve $\gamma$ plays a crucial physical role since it represents the coexistence of two different thermodynamic phases  and the point $(h_c,J_c)$ is the \textit{critical point} of the system, whose exact value is computed in \cite{ACM}.

\begin{figure}[h!]\label{fig.3.3}
\centering
\includegraphics[scale=0.3]{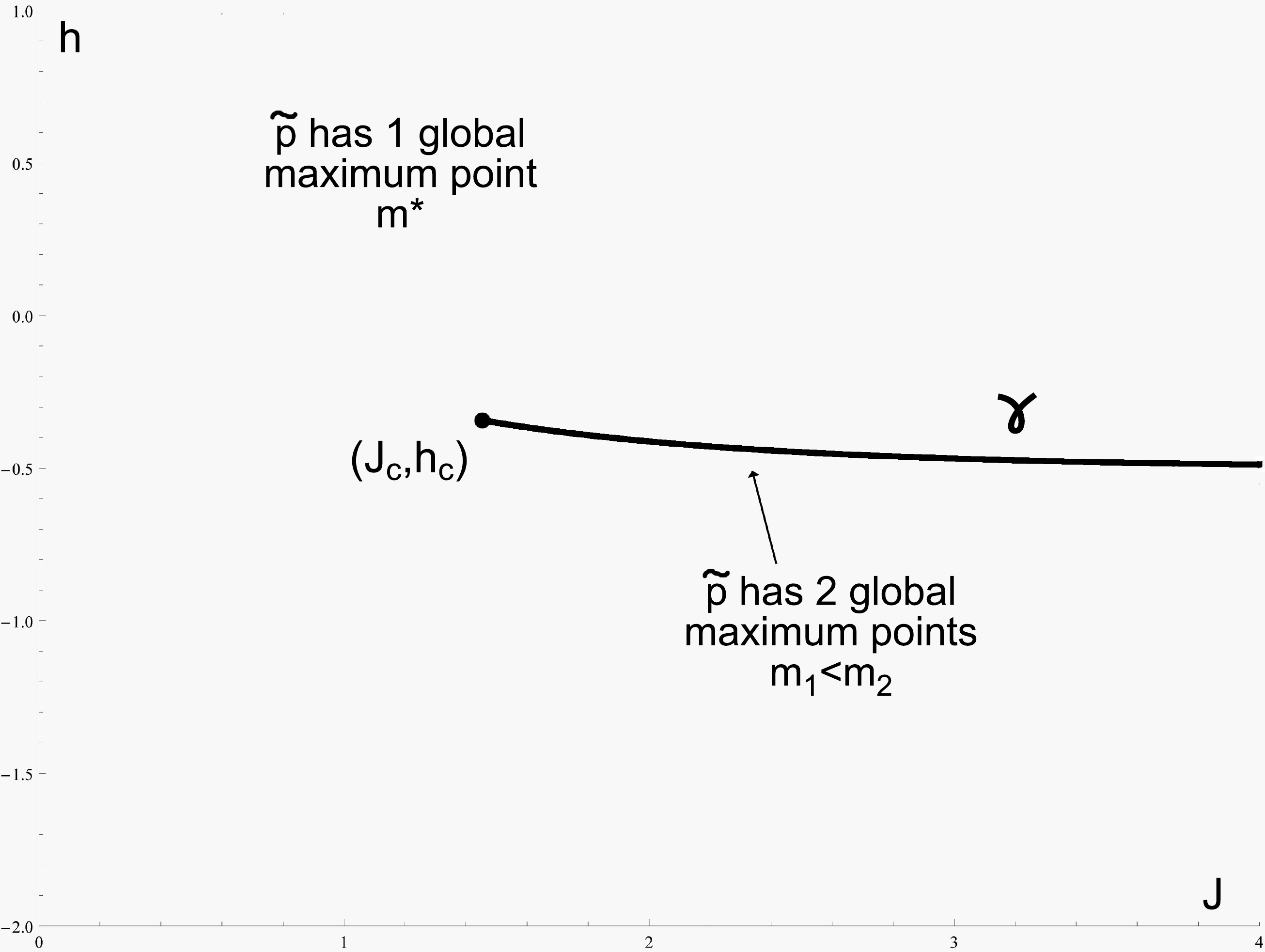}
\caption{The coexistence curve $\gamma$ and the critical point $(h_c,J_c)$  in the plane $(h,J)$.}
\end{figure}

Outside of $\gamma$, by differentiating the expression \eqref{termo_limit} with respect to the external field $h$, one obtains that the value $m^*$ maximizing $\widetilde p$ is the limit of the average monomer density $m_N=S_N/N$ with respect to the Gibbs measure:
\begin{equation}\label{firsmoment}
\lim_{N\rightarrow\infty}\E_{\mu_{N}}(m_{N}) \,=\, \lim_{N\rightarrow\infty}\frac{\partial}{\partial h} p_N \,=\,
\frac{\d}{\d h} \widetilde p(m^*) \,=\, m^*(h,J).
\end{equation}
We want to stress the fact  that in the standard mean-field ferromagnetic model (Curie-Weiss model) the existence of the limiting magnetization on the coexistence curve (zero external field) is achieved by a spin flip symmetry argument, a property that we do not have in the present case.

In the next sections we will prove the law of large numbers, the central limit theorem and their breakdowns, respectively  theorems \ref{teorema.1} and \ref{teorema.2} below, for the distribution of $S_N$ (suitable normalised) with respect to the  Gibbs measure $\mu_N$.

\begin{theorem}\label{teorema.1} Consider the IMD model defined by the Hamiltonian \eqref{hami}.
\begin{itemize}
\item[i)] In the uniqueness region $(h,J)\in\R\times\R^+\setminus\gamma$, we have that
\begin{equation}\label{pequal1}
	m_{N}\overset{\mathcal{D}}{\rightarrow}\delta_{m^*} \;.
\end{equation}
\item[ii)] On the coexistence curve $\gamma$, we have that
\begin{equation}\label{pequal2}
	m_{N}\overset{\mathcal{D}}{\rightarrow} \rho_1\,\delta_{m_{1}}+ \rho_2\,\delta_{m_{2}} \;,
\end{equation}
where $\rho_l=\frac{b_{l}}{b_{1}+b_2}\,$, $b_{l}=(-\lambda_l(2-m_l))^{-1/2}\,$  and
$\lambda_l=\frac{\partial^2}{\partial m^2}\widetilde{p}(m_l)\,$, for $l=1,2$.
\end{itemize}
\end{theorem}

\begin{remark}\label{www}
We notice that, on the contrary of what happens for the Curie-Weiss model, the statistical weights $\rho_{1}$ and $\rho_{2}$  on the coexistence curve  are in general different, furthermore they are not simply given in terms of the second derivative of the variational pressure $\widetilde p\,$.

\begin{figure}[h!]\label{fig.3.4}
\centering
\includegraphics[scale=0.3]{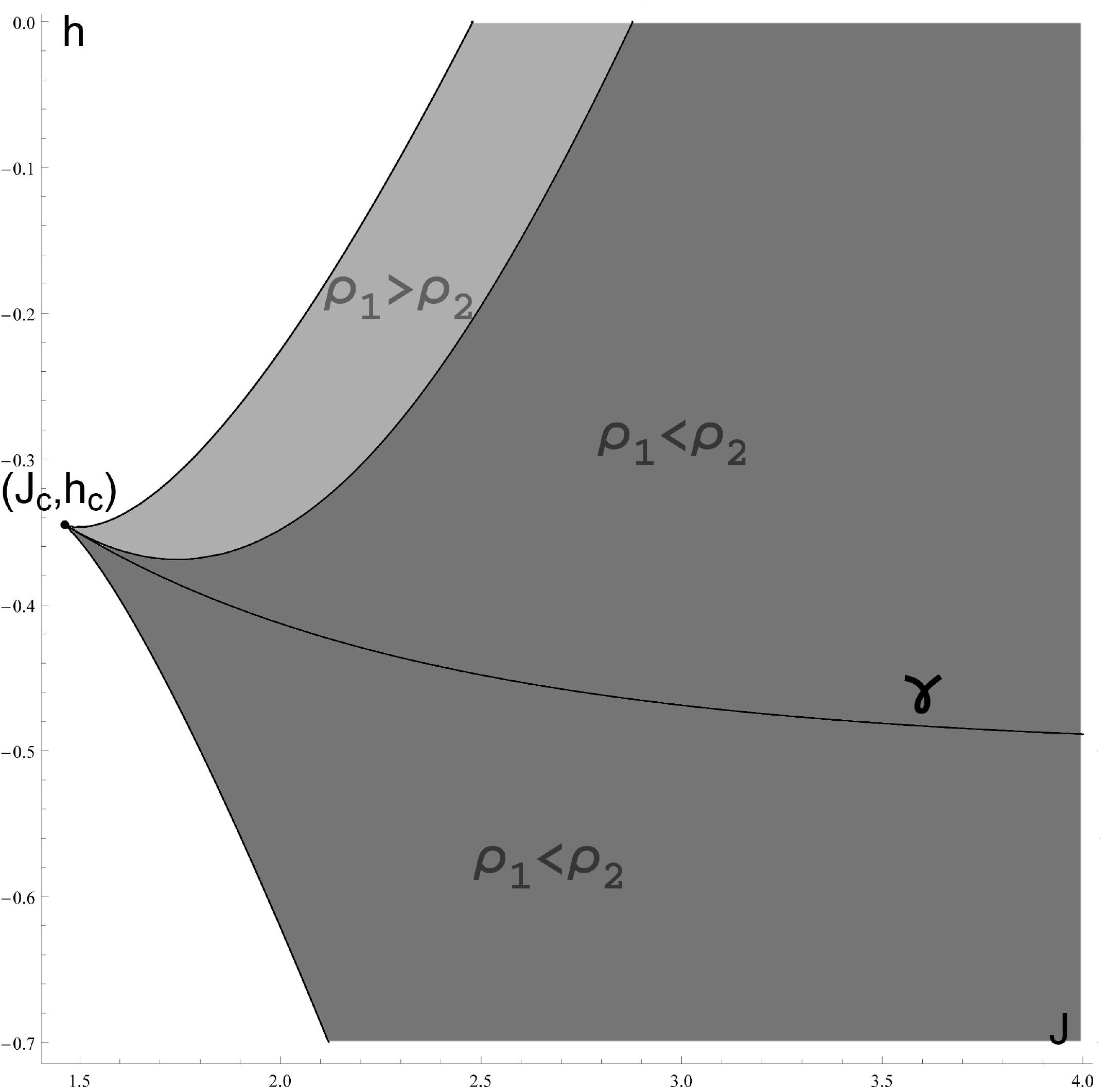}
\caption{ We extend the definition of $\rho_1$ and $\rho_2$ on a region that contains $\gamma$, where $m_1$ and $m_2$ are  local maximum points of $\widetilde{p}$, and then we compute the sign of $\rho_1-\rho_2\,$.  The coexistence curve appears to be  completely contained in the region  $\rho_1 < \rho_2$. }
\end{figure}

The first fact can be seen  numerically (figure \ref{fig.3.4}) and analytically one can compute
\begin{equation}\label{boundgamma}
\lim_{J\to \infty}\frac{\rho_1(J)}{\rho_2(J)}=\frac{1}{\sqrt 2} \;.
\end{equation}
Indeed, by exploiting the formula $(p^{(0)})'' = g' = 2g(1-g)/(2-g)\,$ (see Appendix of \cite{ACM}), one can rewrite the ratio $\rho_1/\rho_2$ as
\eq \label{eq: p''}
\frac{\rho_1}{\rho_2}\,=\,\sqrt{ \frac{(2-m_2) - 4J\,m_2\,(1-m_2)}{(2-m_1) - 4J\,m_1\,(1-m_1)}} \;.
\eeq

Furthermore, the relative weights $\rho_l$ have two contributions reflecting the presence of two different kind of interaction: the first contribution $\lambda_l$ is given by the second derivative of the variational pressure \eqref{tildep}, while the second contribution $2-m_l$ comes from the second derivative of the pressure of the pure hard-core model.
\end{remark}


\begin{theorem}\label{teorema.2} Consider the IMD model defined by the Hamiltonian \eqref{hami}.
\begin{itemize}
\item[i)] For $(h,J)\in\big(\R\times\R^+\big)\setminus\big(\gamma\cup(h_c,J_c)\big)$, we have
\begin{equation}\label{CL1}
\frac{S_{N}-N m^*}{N^{1/2}} \,\overset{\mathcal{D}}{\rightarrow}\, \mathcal{N}\Big(0,\sigma^2\Big)
\end{equation}
where  $\sigma^2=-\lambda^{-1}-(2J)^{-1}>0\,$ and $\lambda=\frac{\partial^2}{\partial m^2}\widetilde{p}(m^*)<0$.
\item[ii)] At the critical point $(h_c,J_c)$, we have
\begin{equation}\label{CLc}
	\frac{S_{N}-N m_c}{N^{3/4}} \,\overset{\mathcal{D}}{\rightarrow}\, C\,\exp\bigg(\dfrac{\lambda_c}{24}\,x^{4}\bigg)dx
	\end{equation}
where $\lambda_c=\frac{\partial^4}{\partial m^4}\widetilde{p}(m_c)<0$, $m_c\equiv m^*(h_c,J_c)$
and $C^{-1}=\int_{\R}\exp (\frac{\lambda_c}{24}x^{4})dx$.
\end{itemize}
\end{theorem}

\section{The pure hard-core model}

A basic ingredient of all the proofs  is the knowledge of the properties of the moment generating function of $S_N$ w.r.t. the Gibbs measure at $J=0$. However, compared with spin models, monomer-dimer models have an additional feature: the problem at $J=0$ is itself non trivial in the sense that the Gibbs measure is not a product measure. We start by deriving the properties of the partition function of the model at $J=0$ that will be used during all the proofs.

For given $u,t\in\R$ and $\eta\geq0$, let us consider
\be\label{momentgen}
\Zcl\big(u+\frac{t}{N^{\eta}}\big) \,=\, \sum_{D\in\DD_N}N^{-|D|}\exp\Big(\big(u+\frac{t}{N^{\eta}}\big)\,S_N(D)\Big)
\ee

In order to obtain an asymptotic expansion of \eqref{momentgen}, which allows us to obtain its scaling properties, we will use a connection between the monomer-dimer problem and Gaussian moments \cite{ACM2,Vlad}. Following the same argument of \cite{ACM2} one finds

\begin{proposition}\label{prop: gaussrepd}
The following representation of the partition function holds
\be\label{gaussrepd}
\Zcl(u+\frac{t}{N^{\eta}}) \,=\, \sqrt{\frac{N}{2\pi}}\, \int_{\mathbb{R}} \Big(\Psi_N(x)\Big)^N dx \;,
\ee
where
\be\label{bigpsi}
\Psi_N(x) \,:=\, \Big(x+\exp(u+\frac{t}{N^{\eta}})\Big)\, \exp(-\frac{x^2}{2}) \;.
\ee
\end{proposition}

%
%
%
%

The above Gaussian representation allows us to use Theorem \ref{th: main} (see the Appendix), an extension of the Laplace method, to obtain a useful asymptotic expansion of $\Zcl(u+\frac{t}{N^{\eta}})$. Precisely
\begin{proposition}\label{asympzeta}
For a given  $u,t\in\R$ and $\eta\geq0$
\be\label{asymzre}
\Zcl\big(u+\frac{t}{N^{\eta}}\big)\equiv\exp\Big(N\, \pcl\big(u+\frac{t}{N^\eta}\big)\Big)\,\underset{N\to\infty}{\sim}\,\exp\Big(N\, \poo\big(u+\frac{t}{N^\eta}\big)\Big)\,\sqrt{\frac{1}{2-g(u)}}
\ee
where  $\poo$ and $g$ are defined respectively in \eqref{funzione.p} and \eqref{funzione.g}.
\end{proposition}

\begin{proof}
Use proposition \ref{prop: gaussrepd} and check that the function $\Psi_N$ defined in \eqref{bigpsi} satisfies the hypothesis of Theorem \ref{th: main}, with $\x_N=e^{-(u+t/N^\eta)}g(u+t/N^\eta)\,$.
By means of the stationarity condition $\x_N^2+e^{u+t/N^\eta}\x_N-1=0$, one finds $\log \Psi_N(\x_N)=\poo(u+t/N^\eta)$ and $\frac{\partial ^2}{\partial x^2} \log \Psi_N(\x_N)=-2+g(u+t/N^\eta)$.\qed
\end{proof}

We will show that the previous proposition gives immediately Theorem \ref{teo1}. On other hand, in the  case $J>0$ we need additional information about the convergence of $\pcl$ to $\poo$.

\begin{proposition}\label{uniformlycon}
For each $k \in \{0, 1, 2,\ldots \}$, $\frac{\partial^k}{\partial h^k}\pcl(h)$ converges uniformly to $\frac{\partial^k}{\partial h^k}\poo(h)$ on compact subsets of $\R$.
\end{proposition}

\begin{proof}
The location of the complex zeros $h\in\C$ of the partition function $\Zcl(h)$ was described in the work of Heilmann and Lieb in \cite{HL}: Theorem 4.2 in \cite{HL} shows that these zeros satisfy $\Re(e^h)=0$, that is $\Im (h)\in\frac{\pi}{2}+\pi\Z\,$.
Set $U:=\R+i\,\big(-\frac{\pi}{4},\frac{\pi}{4}\big)\,\subset\C\,$.
The analytic function $\Zcl(h)$ does not vanish on the simply connected open set $U$, hence $\pcl(h)\equiv \frac{1}{N}\log \Zcl(h)$ is a well-defined analytic function on $U$.
Moreover the sequence $\big(\pcl(h)\big)_{N\in\N}$ is bounded uniformly in $N$ and in $h\in K$, for every $K$ compact subset of $U$; indeed
\[ \big|\pcl(h)\big| \,\leq\, \frac{1}{N}\,\big|\log\big|\Zcl(h)\big|\,\big| \,+\, \frac{2\pi}{N} \;,\]
from the definition of $\Zcl$ it follows immediately
\[ \frac{1}{N}\,\log\big|\Zcl(h)\big| \,\leq\, \frac{1}{N}\log\Zcl\Big(\sup_{h\in K}\Re (h)\Big) \;,\]
and on the other hand, since $\Zcl$ is a polynomial in the variable $e^h$, using the Fundamental Theorem of Algebra and thank to the choice of $U$, it follows
\[ \frac{1}{N}\,\log\big|\Zcl(h)\big| \,\geq\, \inf_{h\in K}e^{\Re (h)}\; \frac{\sqrt2}{2} \;.\]
Thus, the claim is a consequence of the Vitali-Porter and Weiestrass Theorems \cite{schiff}.\qed
\end{proof}

Let now prove {\bf Theorem \ref{teo1}}. For each $u\in\R$ and $\eta\geq0$ we define
\be\label{scalmag0}
S_{N,\eta,u}:= \frac{S_N-u}{N^\eta}.
\ee
In order to prove the two statements of the Theorem \ref{teo1}, namely the law of large numbers \eqref{LLN0} and the central limit theorem \eqref{CLT0}, it is enough to compute the limit of the  moment generating function of $S_{N,\eta,u}$ for $\eta=1,u=0$ and for $\eta=\frac{1}{2},u=g(h)$ respectively.

Consider the moment generating function of $S_{N,\eta,u}$ with respect to the Gibbs measure $\mucl\,$ with external field $h$, namely for all $t\in\R$
\be\label{momentgen0}
\phi_{S_{N,\eta,u}}(t) \,:=\, \sum_{D\in\DD_N}\mucl(D)\; e^{t\,S_{N,\eta,u}(D)} \;.
\ee
By \eqref{momentgen} one can rewrite \eqref{momentgen0} as
\begin{equation}\label{mome}
  \phi_{S_{N,\eta,u}}(t) \,=\, e^{-tu/N^\eta}\; \frac{\Zcl(h+\frac{t}{N^\eta})}{\Zcl(h)} \;.
\end{equation}
Using  proposition \ref{asympzeta} for the numerator and the denominator of \eqref{mome} one gets
\begin{equation}
\dfrac{\Zcl(h+\frac{t}{N^\eta})}
  {\Zcl(h)}\,\underset{N\to\infty}{\sim}\,\exp\Bigg(N\bigg(\poo\big(h+\frac{t}{N^\eta}\big)-\poo(h)\bigg)\Bigg)
\end{equation}

Setting $\eta=1$ and $u=0$ and using the Taylor expansion $\poo(h+\frac{t}{N})-\poo(h)=\frac{t}{N}\frac{\partial}{\partial h}\poo(h)+O(N^{-2})$ and \eqref{gprop},  we obtain
\be\label{MGFLL0}
\lim_{N\to\infty}\phi_{S_{N,1,0}}(t) \,=\, e^{t\,g(h)} \quad\forall\,t\in \R
\ee
which implies \eqref{LLN0}.

In the case of the central limit theorem, setting $\eta=\frac{1}{2}$ and $u=g(h)$, the leading order is provided by the Taylor expansion of $\poo(h+\frac{t}{\sqrt{N}})$ up to the second order
\begin{equation*}
\poo(h+\frac{t}{\sqrt{N}}) \,=\, \poo(h) + \frac{t}{\sqrt{N}}\,\frac{\partial}{\partial h}\poo(h) + \frac{t^2}{2N}\,\frac{\partial^2}{\partial h^2}\poo(h) + O(N^{-\frac{3}{2}}) \;,
\end{equation*}
and then we obtain
\be\label{MGFLL0}
\lim_{N\to\infty}\phi_{S_{N,\frac{1}{2},g(h)}}(t)= e^{\frac{t^2}{2}\,\frac{\partial}{\partial h}g(h)} \quad\forall\,t\in\R
\ee
which implies \eqref{CLT0} and completes the proof. \qed

\section{The model with attractive potential}

The strategy in the case $J>0$ follows the general method of Ellis and Newman \cite{ellis1978limit}, namely, in order to overcome the obstacle of the quadratic term in the interaction, we consider the convolution of the Gibbs measure $\mu_N$ with a suitable Gaussian random variable. Let us start by two simple lemmas.

\begin{lemma}\label{lemma.1}
For all integer $N$, let $W_{N}$ and $Y_{N}$ be two independent random variables. Assume that $W_{N}\overset{\mathcal{D}}{\rightarrow}W$, where
\begin{equation*}
\E\,e^{itW} \neq 0 \quad\forall t\in\R \;.
\end{equation*}
Then $Y_{N}\overset{\mathcal{D}}{\rightarrow}Y$ if and only if $W_{N}+Y_{N}\overset{\mathcal{D}}{\rightarrow}W+Y\,$.
\end{lemma}

\begin{lemma}\label{lemma.2}
Let $W \sim \mathcal{N}(0,(2J)^{-1})$ be a random variable independent of $S_{N}$ for all $N\in\N$. Then given $\eta\geq0$ and $u\in\R$, the distribution of
\be
\frac{W}{N^{1/2 -\eta}}+\frac{S_{N}-N u}{N^{1-\eta}}
\ee
is
\begin{equation}\label{tesi.lemma.2}
C_N\,\exp\Big(N\,F_N\Big(\dfrac{x}{N^{\eta}}+u\Big)\Big)\,dx \;,
\end{equation}
where $C_N^{-1}=\int_{\mathbb{R}}\exp\big(N\,F_N(\frac{x}{N^{\eta}}+u)\big)dx$ and
\begin{equation}\label{funzione.F}
F_N(x):=-Jx^2+\pcl(2Jx+h-J) \;.
\end{equation}
\end{lemma}

\begin{proof}
Given $\theta\in\R\,$,
\be
\Prob\bigg\{\frac{W}{N^{1/2 -\eta}}+\frac{S_{N}-Nm}{N^{1-\eta}}\leq\theta\bigg\} \,=\, \Prob\Big\{\sqrt{N}W+S_{N}\in E\Big\}
\ee
where $E=(-\infty,\theta N^{1-\eta}+Nm]$.
%

The law of $\sqrt{N}W+S_{N}$ is given by the convolution of the Gaussian $\,\mathcal{N}(0,N (2J)^{-1})$
with the distribution of $S_{N}$ w.r.t. the Gibbs measure $\mu_N\,$:
\be\label{pas34}	
\begin{split}
&\Prob\Big\{\sqrt{N}W+S_{N}\in E\Big\}=\\
&\bigg(\frac{ J}{\pi N}\bigg)^{\frac{1}{2}}\int_{E} dt\; \E_{\mu_N}\exp\bigg(-\frac{J}{N}\Big(t-S_N\Big)^{2}\bigg)=\\
&\frac{1}{Z_{N}}\bigg(\frac{ J}{\pi N}\bigg)^{\frac{1}{2}} \int_{E} dt\; \exp\bigg(-\frac{J}{N}t^{2}\bigg)\;
\Zcl\bigg(\frac{2Jt}{N}+h-J \bigg) \;,
\end{split}
\end{equation}
where the last equality follows from  \eqref{momentgen}.
Making the change of variable $x=(t-Nu)/{N^{1-\eta}}$ in \eqref{pas34}, we obtain:
\begin{equation}\label{pim}
P\Big\{\sqrt{N}W+S_{N}\in E\Big\} \,=\, C_N \int_{-\infty}^{\theta}dx\; \exp\bigg(-JN\Big(\frac{x}{N^{\eta}}+u\Big)^{2}\bigg)\; \Zcl\bigg(2J\Big(\frac{x}{N^{\eta}}+u\Big)+h-J\bigg)
\end{equation}
and the integrated function can be rewritten as \eqref{tesi.lemma.2}. \qed
\end{proof}
	
The core of the problem is the convergence of the sequence of measures determined by \eqref{tesi.lemma.2} for suitable values of $\eta$ and $u$.
Thus, we are interested in the limit of quantities of the form
\be\label{above}
\int_{\mathbb{R}}\exp\Big(N\,F_N\Big(\dfrac{x}{N^{\eta}}+u\Big)\Big)\,\psi(x)\,dx
\ee
where $\psi$ is an arbitrary bounded continuous function.
Clearly, the results depend crucially on the scaling properties of $F_N$ near its maximum point(s). By \eqref{funzione.F}, \eqref{tildep} and \eqref{eq: p classic limit} we know that
\be\label{Fennelimit}
\lim_{N \to \infty} F_N(x)=\widetilde{p}(x),\,\,\forall\,x\in\R \;.
\ee
However, the study of the asymptotic behaviour of the integral \eqref{above} requires stronger convergence results provided by propositions \ref{asympzeta} and \ref{uniformlycon}.

Given a sequence of functions $f_N:\R\to\R$, for any $x,y \in\R$   we define
\be\label{deltavariation}
\Delta{f_N} (x;y):={f_N}(x+y)-{f_N}(y).
\ee

Let $\mu\equiv\mu(h,J)$ be a maximum point of $\widetilde p$ and denote by $2k$ the order of the first non zero derivative at $\mu$. Hence, making a Taylor expansion, one finds as $N\to\infty$
\be\label{asypitilde}
N\,\Delta\widetilde{p} (x\,N^{-\frac{1}{2k}};\mu) \,=\, \frac{\lambda}{(2k)!}\, x^{2k} + O\Big({N^{-\frac{1}{2k}}}\Big)
\ee
where $\lambda=\frac{\partial^{2k}}{\partial m^{2k}}\widetilde{p}(\mu)<0$.

The next proposition relates  the asymptotic behaviors of $N\,\Delta{F_N}$ and $N\,\Delta{\widetilde{p}}$.
\begin{proposition}\label{asyequivalence}
For any $x,y\in\R$ and $\eta\geq0$,
\be\label{asyeq}
\lim_{N\to\infty} \exp\Big(N\,\Big(F_N(x\,N^{-\eta}+y)-\widetilde{p}(x\,N^{-\eta}+y)\Big)\Big)\,=\,c(y)
\ee
where $c(y):= \big(2-g(2Jy+h-J)\big)^{-1/2}$. Hence,
\be\label{deltalimit}
N\Big(\Delta{F_N} (x\,N^{-\eta};y)-\Delta{\widetilde{p}} (x\,N^{-\eta};y)\Big) \,\underset{N\to\infty}\rightarrow\, 0 \;.
\ee
\end{proposition}

\begin{proof}
Keeping in mind the definitions \eqref{funzione.F}, \eqref{tildep} and using Proposition \ref{asympzeta} we get \eqref{asyeq}. Then  \eqref{deltalimit} is a straightforward  consequence. \qed
\end{proof}

The next two propositions allow us to control the integral \eqref{above} in the large $N$ limit.

\begin{proposition}\label{cutoff}
Set $M:=\max\{\tilde{p}(x)|x\in\mathbb{R}\}$, let $\mathcal{C}$ be any closed (possibly unbounded) subset of $\mathbb{R}$ which contains no global maximum points of $\tilde{p}$. Then there exists $\eps>0$ such that
\begin{equation}
e^{-NM}\int_{\mathcal{C}} e^{NF_N(x)}dx \,=\, O(e^{-N\eps}) \quad\text{as }N\rightarrow\infty.
\end{equation}
\end{proposition}

\begin{proof}
We observe that the sequence of functions $(\pcl)_{N\in\N}$ is uniformly Lipschitz with constant $1$, namely for all $h,h'\in\R$ and $N\in\N$
\be\label{Lipschitz}
|\pcl(h)-\pcl(h')| \,\leq\, |h-h'| \;,
\ee
since $\frac{\partial}{\partial h} \pcl = \E_{\mucl}(m_{N}) \in[0,1]\,$.
From \eqref{Lipschitz} and definition \eqref{funzione.F} we get
\be\label{Limit}
\lim_{|x|\to\infty}\sup_N F_N(x)=-\infty
\ee
and
\be\label{finiteinte}
\sup_N \int_{\mathbb{R}}e^{F_N(x)} dx<\infty \;.
\ee
Fixed $\eps_1>0$, by \eqref{Limit} we can pick a number $A\in\R$ sufficiently large such that
\be\label{cutoffA}
\sup_{x\in \mathcal{O}_A}F_N(x)-M\leq-\eps_1 \quad\forall\,N\in\N
\ee
where $\mathcal{O}_A\equiv\{x\in\R: |x|> A\}$.
Furthermore $\mathcal{C}\setminus \mathcal{O}_A$ is compact (or possibly empty) and then, by proposition \ref{uniformlycon}, there exist $\eps_2>0$ and $\bar N$ such that
\be\label{cutoffcompact}
\sup_{x\in \mathcal{C}\setminus\mathcal{O}_A} F_N(x)-M \leq -\eps_2 \quad\forall\,N>\bar{N} \;.
\ee
Thus setting $\eps:=\min(\eps_1,\eps_2)$ we get
\be
\sup_{x\in \mathcal{C}}F_N(x)-M\leq -\eps \quad\forall\,N>\bar{N}
\ee
Hence, for $N>\bar{N}$,
\be \begin{split}
e^{-NM} \int_{\mathcal{C}}e^{N\,F_N(x)}dx \,&\leq\,
e^{-NM}\, e^{(N-1)(M-\epsilon)} \int_{\mathcal{C}} e^{F_N(x)}dx \\
&\leq\, e^{-N\eps}\, e^{-(M-\eps)} \int_{\mathbb{R}} e^{F_N(x)}dx \;.
\end{split} \ee
The last is uniformly bounded in $N$  by \eqref{finiteinte} and this completes the proof.\qed
\end{proof}

In the rest of this section  $\partial^{k}f(x)$ denotes the $k^{\mathrm{th}}$-derivative of a function $f$ at the point $x$.

\begin{proposition}\label{prop: DCT}
Let $\mu$ be a maximum point of $\widetilde{p}$, let $2k$ be the order of the first non-zero derivative of $\widetilde p$ at $\mu$. Given $\delta,\eps>0$, there exists $\overline{N}_\eps$ such that for all $N>\overline{N}_{\eps}$
\be \label{DCT}
N\,\Delta F_N\big(x\,N^{-\frac{1}{2k}};\mu\big) \,\leq\, \sum_{j=1}^{2k-1} \eps\, x^j + L_{\delta,\eps}\,x^{2k}
\qquad\forall\,x,\,|x|<\delta N^\frac{1}{2k}
\ee
where
\be \label{remanu}
L_{\delta,\eps} \,:=\, \frac{\partial^{2k}\widetilde{p}(\mu)+\eps}{(2k)!} \,+\,
\delta\;\frac{\sup_{[\mu-\delta,\mu+\delta]}|\partial^{2k+1}\widetilde{p}|+\eps}{(2k+1)!} \;.
\ee
In particular, since $\partial^{2k}\widetilde{p}(\mu)<0$, one can choose $\delta,\eps>0$ such that $L_{\delta,\eps}<0$, and then the sequence of functions
\be \label{dominate}
\exp\big(N\,\Delta {F_N}(x\,N^{-\frac{1}{2k}};\mu)\big)\;\1\big(|x|<\delta N^{\frac{1}{2k}}\big)
\ee
turns out to be dominated by an integrable function of $x\,$.
\end{proposition}

\begin{proof}
The Taylor expansion of $F_N$ at the point $\mu$ gives
\be \label{control}
N\,\Delta{F_N} (x\,N^{-\frac{1}{2k}};\mu) \,=\,
\sum_{j=1}^{2k-1}\frac{\partial^jF_N(\mu)}{j!}\,N^{1-j/2k}\,x^j \,+\,
\frac{\partial^{2k}F_N(\mu)}{(2k)!}\,x^{2k} \,+\,
\frac{\partial^{2k+1}F_N(\xi)}{(2k+1)!}\,N^{-\frac{1}{2k}}\,x^{2k+1}
\ee
where $\xi\in(\mu,\mu+x\,N^{-\frac{1}{2k}})$.
We claim that for any $j\in\{1,\ldots, 2k-1\}$
\be \label{pointwise.der}
\partial^jF_N(\mu)\,N^{1-j/2k} \,\underset{N\to\infty}\rightarrow\,0 \;.
\ee
Indeed, by \eqref{deltalimit}
\be \label{appo}
N\Big(\Delta{F_N} (x\,N^{-\frac{1}{2k}};\mu)-\Delta{\widetilde{p}} (x\,N^{-\frac{1}{2k}};\mu)\Big) \,\underset{N\to\infty}\rightarrow\, 0 \;,
\ee
that is, by substituting \eqref{control} and \eqref{asypitilde} into \eqref{appo},
\be \label{control2}
\sum_{j=1}^{2k-1}\frac{\partial^jF_N(\mu)}{j!}\,N^{1-j/2k}\,x^j \,+\,
\frac{\partial^{2k}F_N(\mu)-\partial^{2k}\widetilde p(\mu)}{(2k)!}\,x^{2k} \,+\,
O\Big(N^{-\frac{1}{2k}}\Big) \,\underset{N\to\infty}\rightarrow\, 0 \;,
\ee
hence using proposition \ref{uniformlycon}, we get
\be
\sum_{j=1}^{2k-1}\frac{\partial^jF_N(\mu)}{j!}\,N^{1-j/2k}\,x^j \,\underset{N\to\infty}\rightarrow\, 0
\ee
which implies \eqref{pointwise.der} since $x$ is arbitrary.
Thus \eqref{pointwise.der} gives the control of the terms of order up to $2k-1$ in \eqref{control}.
The last two terms in \eqref{control} can be grouped together observing that $|x|^{2k+1} < x^{2k}\delta N^\frac{1}{2k}$; then the estimate \eqref{DCT} is obtained using the uniform convergence of $\partial^{2k}F_N$, $\partial^{2k+1}F_N$ on the compact set $[\mu-\delta,\mu+\delta]$, which is guaranteed by proposition \ref{uniformlycon}. \qed
\end{proof}

Let now prove {\bf Theorem \ref{teorema.1}}.
We denote by $\mathcal{M}=\{\mu_l\}_{l=1,\ldots, P}$ the set global maximum points of $\widetilde{p}$
and let $k_l$ and $\lambda_l$ be as in \eqref{asypitilde}.
Set $M := \max_{m} \widetilde{p}(m) = \widetilde{p}(\mu_l)$ for each $l=1,\ldots, P$.
From the analysis  of $\widetilde{p}$ and using the properties of the function $g$ (see \cite{ACM}), it turns out that $k_l$ do not depend on $l$ and precisely
\be\label{typemaximum}
\big(\mathcal{M}, k\big)=
\begin{cases}
\big(\{m^*(h,J)\},1\big) & \text{if }\,(h,J)\in(\R\times\R^+)\setminus\big(\gamma\cup(h_c,J_c)\big)\\
\big(\{m_c\},2\big) & \text{if }\,(h,J)=(h_c,J_c)\\
\big(\{m_1(J),m_2(J)\},1\big) & \text{if }\,(h,J)\in\gamma \,.
\end{cases}
\ee

The argument described below applies in all the cases proving respectively \eqref{pequal1} and  \eqref{pequal2}.
Keeping in mind \eqref{typemaximum}, we proceed with the computation of the limiting distribution of the monomer density $m_N=S_N/N$. By lemmas \ref{lemma.1} and \ref{lemma.2} with $\eta=0$ and $u=0$, it suffices to prove that for any bounded continuous function $\psi$
\begin{equation}\label{teo.1.passaggio.1}
\frac{\displaystyle{\int_{\mathbb{R}}}e^{N\,F_N(x)}\psi(x)dx}{\displaystyle{\int_{\mathbb{R}}}e^{N\,F_N(x)}dx}
\,\rightarrow\, \dfrac{\sum\limits_{l=1}^{P}\psi(\mu_{l})b_{l}}{\sum\limits_{l=1}^{P}\,b_{l}} \;.
\end{equation}
For each $l=1,\dots,P$ let $\delta_l>0$ such that the sequence of functions \eqref{dominate}, with $\mu_l$ in place of $\mu$, is dominated by an integrable function.
We choose $\bar{\delta}=\min\{\delta_{l}\;|\;l=1,\dots,P\}$, decreasing it (if necessary) to assure that $0<\bar{\delta}<\min\{|\mu_{l}-\mu_{s}|:1\leq l\neq s\leq P\}$.
Denote by $\mathcal{C}$ the closed set
\begin{equation*}
\mathcal{C} := \mathbb{R} \setminus \bigcup_{l=1}^{P}(\mu_{l}-\bar{\delta},\mu_{l}+\bar{\delta}) \;;
\end{equation*}
by proposition \ref{cutoff} there exists $\eps>0$ such that as $N\rightarrow\infty$
\begin{equation}\label{teo.1.passaggio.2}
e^{-NM}\int_{\mathcal{C}}e^{N\,F_N(x)}\psi(x)dx \,=\, O(e^{-N\eps}) \;.
\end{equation}
For each $l=1,\dots,P$ we have
\be \label{passaggio} \begin{split}
& N^\frac{1}{2k}\, e^{-NM} \int_{\mu_{l}-\bar{\delta}}^{\mu_{l}+\bar{\delta}}e^{N\,F_N(x)}\psi(x)\,dx \,=\\
&=\, e^{N(F_N(\mu_l)-M)}\int_{|w|<\bar{\delta}N^\frac{1}{2k}}\, \exp\Big(N\,\Delta F_N\big(wN^{-\frac{1}{2k}};\mu_l\big)\Big)\, \psi\big(wN^{-\frac{1}{2k}}+\mu_l\big)\,dw
\end{split} \ee
where the  equality follows from the change of variable $x=\mu_l+wN^{-\frac{1}{2k}}$ and $\Delta F_N$ is defined in \eqref{deltavariation}.

\noindent Since $M\equiv \widetilde{p}(\mu_l)$, from \eqref{asyeq} we know that
\be\label{sorpresa}
\lim_{N\to\infty}e^{N\,(F_N(\mu_l)-M)}=\frac{1}{\sqrt{2-g(2J\mu_l+h-J)}}=\frac{1}{\sqrt{2-\mu_l}}
\ee
where the last equality follows from the fact that $\mu_l$ must satisfy the consistency equation \eqref{consequation}.

\noindent By Proposition \ref{prop: DCT} we can apply the dominated convergence theorem  to the integral on the r.h.s. of \eqref{passaggio}, then by \eqref{deltalimit} and \eqref{asypitilde} we obtain
\be \label{teo.1.passaggio.4} \begin{split}
& \lim_{N\to\infty} N^\frac{1}{2k}\, e^{-NM} \int_{\mu_l-\bar{\delta}}^{\mu_l+\bar{\delta}} e^{N\,F_N(x)} \,\psi(x)\,dx \,=\\
&=\, \frac{1}{\sqrt{2-\mu_l}}\, \int_{\mathbb{R}}\exp\Big(\frac{\lambda_{l}}{(2k)!}w^{2k}\Big)\, \psi(\mu_l)\,dw \;.
\end{split} \ee
Making the change of variable $x=w(-\lambda_{l})^\frac{1}{2k}$ in the r.h.s. of (\ref{teo.1.passaggio.4}) and using (\ref{teo.1.passaggio.2}) we obtain
\begin{equation}\label{teo.1.passaggio.6}
\lim_{N\rightarrow\infty} N^\frac{1}{2k}\, e^{-NM} \int_{\mathbb{R}}e^{N\,F_N(x)}\psi(x)dx \,=\,
\sum_{l=1}^{P} \frac{1}{\sqrt{2-\mu_l}}\,(-\lambda_{l})^{-\frac{1}{2k}}\,\psi(\mu_{l})\, \int_{\mathbb{R}}\exp\Big(-\frac{x^{2k}}{(2k)!}\Big)\,dx \;.
\end{equation}
The analogous limit for the  denominator of \eqref{teo.1.passaggio.1} follows from \eqref{teo.1.passaggio.6} by choosing $\psi=1$.
This concludes the proof of the Theorem \ref{teorema.1}. \qed

Let now prove the {\bf Theorem \ref{teorema.2}}.
Keeping in mind \eqref{typemaximum}, let us start by proving the following
\begin{equation}\label{teo.2.passaggio.1.bu}
\frac{\displaystyle \int_{\mathbb{R}} \exp\Big(N\,F_N\big(xN^{-\frac{1}{2k}}+m^*\big)\Big)\, \psi(x)\, dx}
{\displaystyle \int_{\mathbb{R}} \exp\Big(N\,F_N\big(xN^{-\frac{1}{2k}}+m^*)\big) \,dx} \,\rightarrow\,
\frac{\displaystyle \int_{\mathbb{R}} \exp\Big(\dfrac{\lambda}{(2k)!}\,x^{2k}\Big) \,\psi(x)\,dx}
{\displaystyle \int_{\mathbb{R}} \exp\Big(\dfrac{\lambda}{(2k)!}\,x^{2k}\Big)\,dx}
\end{equation}
for any bounded continuous function $\psi$.
We pick $\delta>0$ such that the sequence of functions \eqref{dominate} is dominated by a integrable function.
By proposition \ref{cutoff} there exists $\eps>0$ such that as $N\rightarrow\infty$
\begin{equation}\label{teo.2.passaggio.2.bu}
e^{-NM} \int_{|x|\geq\delta N^\frac{1}{2k}}\exp\Big(N\,F_N\big(xN^{-\frac{1}{2k}}+m^*\big)\Big) \,\psi(x)\,dx \,=\, O\big(N^\frac{1}{2k}\,e^{-N\eps}\big)
\end{equation}
where $M=\max_m\widetilde{p}(m)$. On the other hand as $|x|<\delta N^{1/2k}$
\be \begin{split}
& e^{-NM} \int_{|x|<\delta N^\frac{1}{2k}} \exp\Big(N\,F_N\big(xN^{-\frac{1}{2k}}+m^*\big)\Big) \,\psi(x)\,dx \,=\\
&=\, e^{(F_N(m^*)-M)}\int_{|x|<\delta N^\frac{1}{2k}}\!\!\exp\Big(N\,\Delta F_N\big(xN^{-\frac{1}{2k}};m^*\big)\Big) \,\psi(x)\,dx \;.
\end{split} \ee
Thus, by proposition \ref{prop: DCT} we can apply the dominated convergence theorem, and then by \eqref{sorpresa}, \eqref{deltalimit} and \eqref{asypitilde} we obtain
\begin{equation}\label{teo.2.passaggio.3.bu}
\lim_{N\rightarrow\infty} e^{-NM} \int_{|x|<\delta N^\frac{1}{2k}} \exp\Big(N\,F_N\big(xN^{-\frac{1}{2k}}+m^*\big)\Big) \,\psi(x)\,dx
\,=\, \frac{1}{\sqrt{2-m^*}}\, \int_{\mathbb{R}} \exp\Big(\frac{\lambda}{(2k)!}\,x^{2k}\Big) \,\psi(x)\,dx
\end{equation}
which, combined with \eqref{teo.2.passaggio.2.bu}, implies \eqref{teo.2.passaggio.1.bu}. 	
	
For $k=2$, by lemmas \ref{lemma.1} and \ref{lemma.2} with $\eta=1/4$ and $u=m^*$, the convergence \eqref{teo.2.passaggio.1.bu} is enough to obtain \eqref{CLc}.

For $k=1$, by lemmas \ref{lemma.1} and \ref{lemma.2} with $\eta=1/2$ and $u=m^*$, since $W\sim\mathcal{N}(0, (2J)^{-1})$, the equation \eqref{teo.2.passaggio.1.bu} implies that the random variable $S_N$ converges to a Gaussian whose variance is $\sigma^2=(-\lambda)^{-1}-(2J)^{-1}$,
provided that
\begin{equation}\label{teo.2.passaggio.4.bu}
(-\lambda)^{-1}-(2J)^{-1}=\dfrac{\lambda+2J}{-2\lambda J}>0
\end{equation}
where  $\lambda=\frac{\partial^2}{\partial m}\widetilde{p}(m^*)\,$. Considering the function $g$ defined in (\ref{funzione.g}), we have that
$\frac{\partial^2}{\partial m}\widetilde{p}(m^*)+2J = (2J)^2\, g'(2Jm^*+h-J)$. Since $g'>0$ and $\lambda<0$ the inequality (\ref{teo.2.passaggio.4.bu}) holds true. \qed

\section*{Acknowledgments} P.C. thanks Chuck Newman for an insightful discussion. D.A. and E.M. thank Giulio Tralli for many interesting discussions. Partial support from FIRB grant RBFR10N90W and from PRIN grant 2010HXAW77 is
acknowledged.

\appendix

\section{Extended Laplace Method}

The usual Laplace method deals with integrals of the form
\[ \int_\R \big(\psi(x)\big)^n\, dx \]
as $n$ goes to infinity. In this appendix we prove a slight extension of the previous method where $\psi$ can depend on $n$. Other results in this direction can be found in \cite{ellisRosen, Laplacerate}.

\begin{theorem}\label{th: main}
For all $n\in\N$ let $\psi_n:\R\to\overline{\R}\,$.
Suppose there exists a compact interval $K\subset\R$ such that $\psi_n>0$ on $K$, so that
\[ \psi_n(x) \,=\, e^{f_n(x)} \quad\forall\,x\in K \;.\]
Suppose that $f_n\in C^2(K)$ and
\begin{itemize}
\item[a)] $f_n \,\underset{n\to\infty}\rightarrow\, f\,$ uniformly on $K\,$;
\item[b)] $f_n'' \,\underset{n\to\infty}\rightarrow\, f''\,$ uniformly on $K\,$.
\end{itemize}
Moreover suppose that:
\begin{itemize}
\item[1)] $\max_K f_n$ is attained in a point $\x_n\in\textrm{int}(K)\,$;
\item[2)] $\limsup_{n\to\infty}\left( \sup_{\R\setminus K}\log|\psi_n| - \max_{K}f_n \right) \,<\, 0\,$;
\item[3)] $\max_K f$ is attained in a unique point $\x\in K\,$;
\item[4)] $f''(\x)<0\,$;
\item[5)] $\limsup_{n\to\infty}\int_\R |\psi_n(x)|\,\d x \,<\, \infty\,$.
\end{itemize}
Then,
\eq \label{eq: main}
\int_\R \big(\psi_n(x)\big)^n\,\d x \,\underset{n\to\infty}{\sim}\, e^{n f_n(\x_n)}\,\sqrt{\frac{2\pi}{-n\,f''(\x)}} \ . \eeq
\end{theorem}

In the proof we use the following elementary fact:
\begin{lemma} \label{lemma}
Let $(f_n)_n$ be a sequence of continuous functions uniformly convergent to $f$ on a compact set $K\,$.
Let $(I_n)_n$ and $I$ be subsets of $K$ such that $\max_{x\in I_n,\,y\in I}\dist(x,y)\to 0$ as $n\to\infty\,$.
Then
\begin{itemize}
\item[$\bullet$] $\max_{I_n}f_n \,\underset{n\to\infty}{\rightarrow}\, \max_{I}f \,$;
\item[$\bullet$] $\argmax_{I_n}f_n \,\underset{n\to\infty}{\rightarrow}\, \argmax_{I}f\,$, provided that $f$ has a unique global maximum point on $I\,$.
\end{itemize}
\end{lemma}

\proof[of the Theorem \ref{th: main}]
Since $\xn$ is an internal maximum point for $f_n$ \textit{(hypothesis 1)}, $f_n'(\xn)=0$ and for all $x\in K$
\eq \label{eq: taylor}
f_n(x) \,=\, f_n(\xn) + \frac{1}{2}\,f''_n(\xi_{x,n})\,(x-\xn)^2 \quad\text{with }\xi_{x,n}\in(\xn,x)\subset K \;.\eeq
Fix $\eps>0$.
Since $f_n''\underset{n\to\infty}{\rightarrow}f''$ uniformly on $K$, there exists $N_\eps$ such that
\eq \label{eq: 1}
|f_n''(\xi) - f''(\xi)| \,<\, \eps \quad\forall\,\xi\in K \quad\forall\,n> N_\eps \;.\eeq
Since $f''$ is continuous in $\x$, there exists $\delta_\eps>0$ such that $B(\x,\delta_\eps)\subset K$ and
\eq \label{eq: 2}
|f''(\xi) - f''(\x)| \,<\, \eps \quad\forall\,\xi:\,|\xi-\x|<\delta_\eps \;.\eeq
By the lemma \ref{lemma} $\xn\underset{n\to\infty}{\rightarrow}\x$, because $\x$ is the unique maximum point of $f$ on $K$ \textit{(hypothesis 3)}. Thus there exists $\bar N_{\delta_\eps}$ such that
\eq \label{eq: 3}
|\xn-\x| \,<\, \frac{\delta_\eps}{2} \quad\forall\,N>\bar N_{\delta_\eps} \;.\eeq
Therefore for $n>N_\eps\lor\bar N_{\delta_\eps}$ and $x\in B(\x,\delta_\eps)$ it holds:
\[\begin{split}
& |\xi_{x,n}-\x| \,\leq\, |\xi_{x,n}-x|+|x-\x| \,\leq\, |\xn-x|+|x-\x| \,\overset{\eqref{eq: 3}}{<}\, \frac{\delta_\eps}{2}+\frac{\delta_\eps}{2} = \delta_\eps \ \Rightarrow\\
& |f_n''(\xi_{x,n})-f''(\x)| \,\leq\, |f_n''(\xi_{x,n})-f''(\xi_{x,n})| + |f''(\xi_{x,n})-f''(\x)| \overset{\eqref{eq: 1},\eqref{eq: 2}}{<} \eps + \eps = 2\eps \;.
\end{split}\]
By substituting into \eqref{eq: taylor} we obtain that for $n>N_\eps\lor\bar N_{\delta_\eps}$ and $x\in B(\x,\delta_\eps)$
\eq \label{eq: taylor ineq}
f_n(x) \begin{cases} \,\leq\, f_n(\xn)+\frac{1}{2}\,\big(f''(\x)+2\eps\big)\,(x-\xn)^2 \\[4pt]
\,\geq\, f_n(\xn)+\frac{1}{2}\,\big(f''(\x)-2\eps\big)\,(x-\xn)^2 \end{cases} \;.
\eeq
Now split the integral into two parts:
\eq \label{eq: split}
\int_\R \big(\psi_n(x)\big)^n\,\d x \,=\, \int_{B(\xn,\delta_\eps)}e^{nf_n(\xn)}\,\d x \,+\, \int_{\R\setminus B(\xn,\delta_\eps)}\,\big(\psi_n(x)\big)^n\,\d x \;.
\eeq

$\bullet$ To control the second integral on the r.h.s. of \eqref{eq: split} we claim that there exists $\eta_{\delta_\eps}>0$ and $N^*_{\delta_\eps}$ such that
\eq \label{eq: eta}
\log|\psi_n(x)| < f_n(\xn) -\eta_{\delta_\eps} \quad\forall\,x\in\R\setminus B(\xn,\delta_\eps) \quad\forall\,N>N^*_{\delta_\eps} \;;
\eeq
namely $\limsup_{n\to\infty} \sup_{x\in\R\setminus B(\xn,\delta_\eps)} \log|\psi_n(x)|-f_n(\xn) < 0\,$.
Indeed:
\[\begin{split}
& \limsup_{n\to\infty} \sup_{x\in\R\setminus B(\xn,\delta_\eps)} \log|\psi_n(x)|-f_n(\xn) \,= \\
& \left(\limsup_{n\to\infty} \sup_{x\in K\setminus B(\xn,\delta_\eps)} f_n(x)-f_n(\xn) \right) \lor
\left(\limsup_{n\to\infty} \sup_{x\in \R\setminus K} \log|\psi_n(x)|-f_n(\xn) \right) \,=\\
& \left(\sup_{x\in K\setminus B(\x,\delta_\eps)} f(x)-f(\x) \right) \lor
\left(\limsup_{n\to\infty} \sup_{x\in \R\setminus K} \log|\psi_n(x)|-f_n(\xn) \right)
\end{split}\]
where the last identity holds true by the lemma \ref{lemma}.\\
Moreover $\sup_{x\in K\setminus B(\x,\delta_\eps)} f(x)-f(\x)<0$ since $\x$ is the unique maximum point of the continuous function $f$ on the compact set $K$ \textit{(hypothesis 3)}; while $\limsup_{n\to\infty} \sup_{x\in \R\setminus K} \log|\psi_n(x)|-f_n(\xn)<0$ by the \textit{hypothesis 2}. This proves the claim.\\
Now using \eqref{eq: eta} and the \textit{hypothesis 5}, there exist $C$ and $N$ such that for all $n>N\lor N^*_{\delta_\eps}$
\eq \begin{split} \label{eq: second integral}
\int_{\R\setminus B(\xn,\delta_\eps)} \big|\psi_n(x)\big|^n\,\d x \,&\leq\,
e^{(n-1)\left(f_n(\x_n)-\eta_{\delta_\eps}\right)}\, \int_\R |\psi_n(x)| \,\d x \\
&\leq\, C\, e^{n\left(f_n(\x_n)-\eta_{\delta_\eps}\right)} \;.
\end{split} \eeq

$\bullet$ To study the first integral on the r.h.s. of \eqref{eq: split}, choose $\eps\in(0,\eps_0]$, where $f''(\x)+2\eps_0<0$ \textit{(hypothesis 4)}.
By \eqref{eq: taylor ineq}, since we can compute Gaussian integrals, we find an upper bound:
\eq \label{eq: ub} \begin{split}
\int_{B(\xn,\delta_\eps)}e^{nf_n(x)}\,\d x \,&\leq\,
e^{nf_n(\xn)}\, \int_\R e^{\frac{n}{2}\,(f''(\x)+2\eps)\,(x-\xn)^2}\,\d x \\
&=\, e^{nf_n(\xn)}\, \frac{1}{\sqrt{-\frac{n}{2}\,(f''(\x)+2\eps)}}\,\int_\R e^{-x^2}\,\d x \\
&=\, e^{nf_n(\xn)}\, \sqrt{\frac{2\pi}{-n\,(f''(\x)+2\eps)}}
\end{split} \eeq
and a lower bound:
\eq \label{eq: lb} \begin{split}
\int_{B(\xn,\delta_\eps)}e^{nf_n(x)}\,\d x \,&\geq\,
e^{nf_n(\xn)}\, \int_{B(\xn,\delta_\eps)} e^{\frac{n}{2}\,(f''(\x)+2\eps)\,(x-\xn)^2}\,\d x \\
&=\, e^{nf_n(\xn)}\, \frac{1}{\sqrt{-\frac{n}{2}\,(f''(\x)+2\eps)}}\,\int_{B\big(0,\,\delta_\eps\sqrt{-\frac{n}{2}(f''(\x)-2\eps)}\big)} e^{-x^2}\,\d x \\
&=\, e^{nf_n(\xn)}\, \sqrt{\frac{2\pi}{-n\,(f''(\x)-2\eps)}}\;(1+\omega_{n,\eps,\delta_\eps})
\end{split} \eeq
where $\omega_{n,\eps,\delta_\eps}\to 0$ as $n\to\infty\,$ and $\eps$ is fixed.

In conclusion, by \eqref{eq: split}, \eqref{eq: second integral}, \eqref{eq: ub}, \eqref{eq: lb} we obtain that for $\eps\in(0,\eps_0]$ and $n>N_\eps\lor\bar N_{\delta_\eps}\lor N\lor N^*_{\delta_\eps}$ it holds:
\[\begin{split}
\frac{\int_\R \big(\psi_n(x)\big)^n\,\d x}{e^{nf_n(\xn)}\sqrt{\frac{2\pi}{-n\,f''(\x)}}} \;&\leq\;
\sqrt{\frac{f''(\x)}{f''(\x)+2\eps}} \,+\, C\,\sqrt{-\frac{n\,f''(\x)}{2\pi}}\,e^{-n\,\eta_{\delta_\eps}} \\
&\underset{n\to\infty}{\rightarrow}\, \sqrt{\frac{f''(\x)}{f''(\x)+2\eps}}
\,\underset{\eps\to0}{\rightarrow}\, 1 \;;
\end{split}\]
and:
\[\begin{split}
\frac{\int_\R \big(\psi_n(x)\big)^n\,\d x}{e^{nf_n(\xn)}\sqrt{\frac{2\pi}{-n\,f''(\x)}}} \;&\geq\;
\sqrt{\frac{f''(\x)}{f''(\x)-2\eps}}\,(1+\omega_{n,\eps,\delta_\eps}) \,-\, C\,\sqrt{-\frac{n\,f''(\x)}{2\pi}}\,e^{-n\,\eta_{\delta_\eps}} \\
&\underset{n\to\infty}{\rightarrow}\, \sqrt{\frac{f''(\x)}{f''(\x)-2\eps}}
\,\underset{\eps\to0}{\rightarrow}\, 1 \;;
\end{split}\]
hence \eqref{eq: main} is proved.\qed
\endproof

\end{document}